\documentclass[11pt]{article}

\usepackage[utf8]{inputenc}
\usepackage[T1]{fontenc}
\usepackage{lmodern}
\usepackage[left=1in,right=1in,top=1in,bottom=1.25in]{geometry}
\usepackage{microtype}
\usepackage{amsmath}
\usepackage{graphicx}
\usepackage{amssymb}
\usepackage{amsthm}
\usepackage{thmtools,xspace}
\usepackage{thm-restate}
\usepackage{algorithm,algorithmic}
\usepackage{titlesec}
\usepackage{fancybox}
\usepackage{xcolor}
\usepackage{xspace}
\usepackage{enumerate}

\usepackage{comment}

\usepackage[
	style=alphabetic,
	backref=true,
	doi=false,
	url=false,
	maxcitenames=3,
	mincitenames=3,
	maxbibnames=10,
	minbibnames=10,
	backend=bibtex8,
	sortlocale=en_US
]{biblatex}

\usepackage[ocgcolorlinks]{hyperref} % Option ocgcolorlinks makes links only colored on screen and not on printouts
\usepackage{cleveref}

\makeatletter
\newcommand*{\rom}[1]{\expandafter\@slowromancap\romannumeral #1@}
\makeatother

\newcommand{\hide}[1]{} 
\newtheorem{claim}{Claim}
\newtheorem{lemma}{Lemma}
\newtheorem{theorem}{Theorem}
\newtheorem{definition}{Definition}

\newcommand{\bin}[2]{\text{Bin}\left(#1,#2\right)}

\newcommand{\KL}[2]{D_{\text{KL}}({#1}||{#2} )}

\newcommand{\pr}[1]{\ensuremath{{\bf{Pr}}\left[{#1}\right]}}

\newcommand{\cc}{\text{Correlation Clustering}\xspace}
\newcommand{\cccc}{\text{2-Correlation-Clustering}\xspace}
\newcommand{\f}{\tilde{f}}
\def\e{\epsilon}
\def\hT{\widehat{T}}
 \def\r{\rho}
\newcommand{\diam}{\frac{ \log{n}}{\log{\log{n}}}}
\def\hD{\widehat{D}}
\def\g{\gamma}
\newcommand{\beql}[1]{\begin{equation}\label{#1}}

\newcommand{\beq}[1]{\begin{equation}\label{#1}}
\newcommand{\eeq}{\end{equation}}

\newcommand{\bfrac}[2]{\left(\frac{#1}{#2}\right)}
\newcommand{\brac}[1]{\left(#1\right)}
\def\a{\alpha}

\newcommand{\field}[1]{\mathbb{#1}} % requires amsfonts

\newcommand{\Prob}[1]{\ensuremath{{\bf{Pr}}\left[{#1}\right]}}
\newcommand{\Mean}[1]{\ensuremath{{\mathbb E}\left[{#1}\right]}}

\newcommand{\whp}{\textit{whp}\xspace}

\newcommand{\spara}[1]{\smallskip\noindent{\bf #1}}

\makeatletter

\def\moverlay{\mathpalette\mov@rlay}
\def\mov@rlay#1#2{\leavevmode\vtop{%
		\baselineskip\z@skip \lineskiplimit-\maxdimen
		\ialign{\hfil$\m@th#1##$\hfil\cr#2\crcr}}}
\newcommand{\charfusion}[3][\mathord]{
	#1{\ifx#1\mathop\vphantom{#2}\fi
		\mathpalette\mov@rlay{#2\cr#3}
	}
	\ifx#1\mathop\expandafter\displaylimits\fi}
\makeatother

\colorlet{DarkRed}{red!50!black}
\colorlet{DarkGreen}{green!50!black}
\colorlet{DarkBlue}{blue!50!black}

\hypersetup{
	linkcolor = DarkRed,
	citecolor = DarkGreen,
	urlcolor = DarkBlue,
	bookmarksnumbered = true,
	linktocpage = true
}

\crefname{algocf}{Procedure}{Procedures}
\Crefname{algocf}{Procedure}{Procedures}

\newenvironment{fminipage}%
  {\begin{Sbox}\begin{minipage}}%
  {\end{minipage}\end{Sbox}\fbox{\TheSbox}}

\title{Predicting Signed Edges with $O(n^{1+o(1)} \log{n})$ Queries}

\author{Michael Mitzenmacher\thanks{Harvard University, \href{mailto:michaelm@eecs.harvard.edu}{michaelm@eecs.harvard.edu}}
\and
Charalampos E. Tsourakakis\thanks{Harvard University, \href{mailto:babis@seas.harvard.edu}{babis@seas.harvard.edu} }
}
 
\date{}

\hypersetup{
	pdftitle = {Predicting Signed Edges with $O(n^{1+o(1)} \log{n})$ Queries},
	pdfauthor = {Michael Mitzenmacher, Charalampos E. Tsourakakis}
}

\begin{document}
\pagenumbering{roman}
\maketitle
\begin{abstract} 
Social networks and interactions in social media involve both positive
and negative relationships. Signed graphs capture both types of
relationships: positive edges correspond to pairs of ``friends'', and
negative edges to pairs of ``foes''.  The {\em edge sign prediction
problem}, which aims to predict whether an interaction between a pair of nodes will be positive or negative,
is an important graph mining task for which many heuristics
have recently been
proposed \cite{leskovec2010predicting,leskovec2010signed}.

Motivated by social balance theory, we model the edge sign prediction problem as a noisy
correlation clustering problem with two clusters. We are allowed to
query each pair of nodes whether they belong to the same cluster or not,
but the answer to the query is corrupted with some probability
$0<q<\frac{1}{2}$. Let $c=\frac{1}{2}-q$ be the gap. We provide an algorithm that recovers the clustering with high probability  in the presence of noise for any constant gap $c$ with 
$O(n^{1+\tfrac{1}{\log\log{n}}}\log{n})$ queries. Our algorithm uses simple breadth first search as its main algorithmic primitive. Finally, we provide a novel generalization to $k \geq 3$ clusters and prove that our techniques can recover the clustering if the gap is constant in this generalized setting.

\end{abstract}

\setcounter{page}{0}
\pagenumbering{arabic}

\section{Introduction}
\label{sec:introduction}
With the rise of social media, where both positive and negative interactions take place, signed graphs,
whose study was initiated by Heider, Cartwright, and Harary \cite{cartwright1956structural,heider1946attitudes,harary1953notion},
have become prevalent in graph mining.  A key graph mining problem is the {\em edge sign prediction problem}, which aims to predict whether an interaction between a pair of nodes will be positive or negative \cite{leskovec2010predicting,leskovec2010signed}. Recent works have developed numerous heuristics for this task that perform relatively well in practice \cite{leskovec2010predicting,leskovec2010signed}. 

In this work we propose a theoretical model for the edge sign prediction problem that highlights its intimate connections with the famous planted partition problem \cite{abbe2016exact,condon2001algorithms,hajek2016achieving,mcsherry2001spectral}.     
Specifically, we model the edge sign prediction problem as a noisy correlation clustering problem, where we are able to query a pair of nodes $(u,v)$ to test whether they belong to the same cluster (edge sign $f(u,v)=+1$) or not (edge sign $f(u,v)=-1$).  The query fails to return the correct answer with some probability $0<q<\frac{1}{2}$.  Correlation clustering is a basic data mining primitive with a large number of applications ranging from social network analysis \cite{harary1953notion,leskovec2010predicting} to computational biology \cite{hou2016new}.  Our theoretical model is inspired by the famous {\em balance theory}: ``the friend of my enemy is my friend'' \cite{cartwright1956structural,easley2010networks,heider1946attitudes}.
The details of our model follow. 

\spara{Model \rom{1}.} Let $V=[n]$ be the set of $n$ items that belong to two clusters, call them red and blue. Set $f:V \rightarrow \{\text{red},\text{blue}\}$, $R = \{v \in V(G): f(v) = \text{red} \}$ and $B = \{v \in V(G): f(v) = \text{blue} \}$, where  $0 \leq |R| \leq n$.  The function $f$ is unknown and we wish to recover the two clusters $R,B$  by querying pairs of items. (We need not recover the labels, just the clusters.)  For each query we receive the correct answer with probability $1-q$, where $q>0$ is the corruption probability. That is, for a pair of items $u,v$ such that $f(u)=f(v)$, with probability $q$ it is reported that  $\tilde{f}(u) \neq \tilde{f}(v)$, and similarly if $f(u) \neq f(v)$ with probability $q$ it is reported that $\tilde{f}(u) = \tilde{f}(v)$.  Our goal is to perform as few queries as possible while recovering the underlying cluster structure. 

\spara{Main result.} Our main theoretical result is that we can recover the clusters $(R,B)$ with high probability\footnote{An event $A_n$ holds with high probability ({\it whp}) 
if $\lim\limits_{n \rightarrow +\infty} \Prob{A_n}=1$.} in polynomial time. Our algorithm uses breadth first search (BFS) as its main algorithmic primitive. Our result is stated as Theorem~\ref{thm:thrm1}.  

\begin{theorem}
\label{thm:thrm1} 
There exists a polynomial time algorithm that performs $\Theta(n \log{n} (2c)^{-\diam})$ edge queries and recovers the clustering $(R,B)$ \whp for any gap $0< c = \frac{1}{2}-q<\frac{1}{2}$.  
\end{theorem}

\noindent  When $c$ is constant, then the number of queries is 
$O(n^{1+\tfrac{1}{\log\log{n}}}\log{n})$.  A natural follow-up question that we address here is whether our results generalize to the case of more than two clusters. We provide a general model and show that our techniques recover the cluster structure \whp as long as the gap $c$ is constant. 

\spara{Model \rom{2}.}  Our model is now that there are $k$ groups, that we number
$\{0,1,...,k-1\}$ and that we think of as being arranged modulo $k$.
Let $g(u)$ refer to the group number associated with a vertex $u$.  We
start by noting that if when querying an edge we returned only whether
the the groups of the two edges were equal, it would be difficult to
reconstruct the clusters; indeed, even with no errors, a chain of such responses
along a path would not generally allow us to determine whether the endpoints of a
path were in the same group or not.  A model that provides
more information and naturally generalizes the two cluster case is the
following:  
when we query an edge $e = (x,y)$, we obtain 

\begin{equation} 
  \label{eq:model2}
    \tilde{f}(e) = \left\{\begin{array}{lr}
        g(x)-g(y) \bmod k , & \text{with probability } 1-q;\\
        g(x)-g(y)+1 \bmod k, & \text{with probability } q/2;\\
        g(x)-g(y)-1 \bmod k, & \text{with probability } q/2.
        \end{array}\right.
\end{equation} 

That is, we obtain the difference between the groups when no error
occurs, and with probability $q$ we obtain an error that adds or subtracts
one to this gap with equal probability.  When $q= 0$, so there are 
no errors from $\tilde{f}(e)$, the edge queries
would allow us to determine the difference between the group numbers of vertices
at the start and end of any path, and in particular would allow us
to determine if the groups were the same.  
We also note that we choose this description
for ease of exposition.  More generally we could handle queries governed by more
general error models, of the form:
$$\tilde{f}(e) = g(x)-g(y)+i \text{~~~with probability~} q_i, 0 \leq i < k.$$
That is, the error does not depend on the group values $x$ and $y$, but is
simply independent and identically distributed over the values $0$ to $k-1$.  

\begin{theorem}
\label{thm:thrm2} 
There exists a polynomial time algorithm that performs $O(n^{1+\tfrac{1}{\log\log{n}}}\log{n})$ edge queries and recovers the $k$ clusters under the model of equation~\eqref{eq:model2}  \whp for any constant gap $0<c<\frac{1}{2}$.
\end{theorem}

Our proof techniques extend naturally to this model. 

\spara{Roadmap.}  Section~\ref{sec:prelim} presents some theoretical preliminaries. Section~\ref{sec:proposed} presents our algorithmic contributions. Section~\ref{sec:related}  briefly  reviews related work. Finally, Section~\ref{sec:concl} concludes the paper.

\section{Theoretical Preliminaries} 
\label{sec:prelim} 
We use the following powerful probabilistic results for the proofs in Section~\ref{sec:proposed}. 

\hide{ 
\begin{theorem}[Union Bound] 
For a countable set of events $A_1,A_2,A_3,\ldots$ 

$$ \Prob{ \cup_{i} A_i } \leq \sum_{i} A_i.$$ 
\end{theorem}
}

\begin{theorem}[Chernoff bound, Theorem 2.1  \cite{janson2011random}]\label{chernoff}
  Let $X\sim\bin{n}{p}$, $\mu=np$, $a\ge 0$ and $\varphi(x)=(1+x)\ln(1+x)-x$
  (for $x\ge -1$, or $\infty$ otherwise). Then the following inequalities
  hold:
  
    \begin{eqnarray}
    \label{chernoff_ineq_low}
    \pr{X\le \mu - a} &\le& e^{-\mu\varphi\left(\frac{-a}{\mu}\right)}
    \le e^{-\frac{a^2}{2\mu}},\\
    \label{chernoff_ineq_high}
    \pr{X\ge \mu + a} &\le& e^{-\mu\varphi\left(\frac{-a}{\mu}\right)}
    \le e^{ -\frac{a^2}{2(\mu+a/3)}}.
  \end{eqnarray} 
\end{theorem}

\noindent We define the notion of read-$k$ families, a useful concept when proving concentration results for weakly dependent variables. 
 
\vspace{0.4cm}
\begin{definition}[Read-$k$ families]
Let $X_1, \dots,   X_m$ be independent random variables. For $j \in [r]$, let $P_j \subseteq [m]$ and let $f_j$ be a Boolean function of $\{ X_i \}_{i\in P_j}$. Assume that 
$ |\{ j | i \in P_j \}| \leq k$  for every $i\in [m]$. Then, the random variables $Y_j=f_j(\{ X_i \}_{i\in P_j})$ are called a read-$k$ family.
\end{definition} 

\vspace{0.4cm}
\begin{theorem}[Concentration of Read-$k$ families \cite{gavinsky2014tail}] 
\label{thm:readk}
Let $Y_1,\ldots,Y_r$ be a family of read-$k$ indicator variables with $\Prob{Y_i=1}=q$. Then for any $\epsilon>0$,
\beql{readk-upper-bound1}
\Prob{ \sum_{i=1}^r Y_i \geq (q+\epsilon) r } \leq e^{-\KL{q+\epsilon}{q} \cdot r/k} 
\eeq 
%\leq e^{-2\epsilon^2 \frac{r}{k} }
and
\beql{readk-lower-bound}
\Prob{ \sum_{i=1}^r Y_i\leq (q-\epsilon) r} \le e^{-\KL{q-\epsilon}{q} \cdot r/k}.
\eeq
%\leq e^{-2\epsilon^2 \frac{r}{k} }.
\end{theorem}

\noindent Here, $D_{\text{KL}}$ is Kullback-Leibler divergence defined as 

$$ \KL{q}{p} = q \log{ \left( \frac{q}{p} \right) }+ (1-q) \log{ \left( \frac{1-q}{1-p} \right) }.$$

We will use the following corollary of Theorem~\ref{thm:readk}, which provides Chernoff-type bounds for read-$k$ families. This is derived in a similar way that Chernoff multiplicative bounds are derived from Equations~\eqref{chernoff_ineq_high} and ~\eqref{chernoff_ineq_low}, see \cite{mcdiarmid1998concentration}. Notice that the main difference compared to the standard Chernoff bounds 
is the extra  $k$ factor in denominator of  the exponent. 
 
\begin{theorem}[Concentration of Read-$k$ families \cite{gavinsky2014tail}] 
\label{thm:readk}
Let $Y_1,\ldots,Y_r$ be a family of read-$k$ indicator variables with $\Prob{Y_i=1}=q$. Also, let $Y=\sum_{i=1}^r Y_i$. Then for any $\epsilon>0$,
\beql{readk-upper-bound}
\Prob{  Y \geq  (1+\epsilon) \Mean{Y} } \leq e^{-\frac{\epsilon^2 \Mean{Y}}{2k(1+\epsilon /3)} }
\eeq  

\beql{readk-lower-bound}
\Prob{  Y \leq  (1- \epsilon) \Mean{Y} } \leq e^{-\frac{\epsilon^2 \Mean{Y}}{2k} }.
\eeq  
\end{theorem}

\section{Proposed Method}
\label{sec:proposed}
We prove our main result through a sequence of claims and lemmas. For completeness we include all proofs even if some claims are classic, e.g., Claim~\ref{claim1}. At a high level, our proof strategy is as follows: 

\begin{enumerate} 
\item We compute the probability that a simple path between $u$ and $v$ provides us with the correct information on whether $f(u)=f(v)$ or not.  
\medskip
\item	 Let $L = \diam$. We show that there exist $N=(2c)^{-L} e^{\tfrac{4L}{5}}$ {\em almost edge-disjoint paths} of length $(1+o(1))L$ between any pair of vertices with probability at least $1-\frac{1}{n^3}$.    The reader can think of the paths as being edge-disjoint, if that is helpful; we shall clarify both what we mean by 
{\em almost edge-disjoint paths} and how it affects the proof later in the paper. 
\medskip
\item  For each path from the collection of $N$  almost edge-disjoint paths, we compute the product of the sign of the edges along the path. Since the paths are not entirely edge disjoint, the corresponding random variables are weakly dependent. We use concentration of multivariate polynomials \cite{gavinsky2014tail}, see also \cite{alonspencer,kim-vu}, in combination with Claim~\ref{claim1} to show that using the majority of the $N$ resulting signs to decide whether $f(u)=f(v)$ or not for a pair of nodes $u,v \in V(G)$ gives the correct answer with probability lower bounded by $1-\frac{1}{n^3}$. Taking the union bound over ${n \choose 2}$ pairs concludes the proof. 
\medskip
\end{enumerate}

The pseudo-code is shown as Algorithm~\ref{alg:2cc}.  The algorithm runs over each pair of nodes, and it invokes Algorithm~\ref{alg:edgeDisjointPaths} to construct  almost edge-disjoint paths for each pair of nodes $u,v$ using Breadth First Search.   Note that since we perform $20n\log{n}(2c)^{-L}$ queries uniformly at random, the resulting graph is  is asymptotically equivalent to $G \sim G(n,\frac{40\log{n}(2c)^{-L}}{n})$, see \cite[Chapter 1]{frieze2015introduction}. Here, $G(n,p)$ is the classic Erd\"{o}s-R\'{e}nyi model  (a.k.a random binomial graph model) where each possible edge between each pair  $(u,v) \in {[n] \choose 2}$ is included in the graph with probability $p$ independent from every other edge.  

It turns out that our algorithm needs an average degree  $O\left (\frac{\log n}{(2c)^L} \right )$ {\em only for the first level} of the trees $T_u,T_v$ that we grow from $u$ and $v$ when we invoke Algorithm~\ref{alg:edgeDisjointPaths}. For all other levels of the grown trees, we need the degree to be only $O(\log{n})$. This difference in the branching factors exists in order to ensure that the number of leaves of trees $T_u,T_v$ in  Algorithm~\ref{alg:edgeDisjointPaths} is amplified by a factor of $\frac{1}{(2c)^L}$, which then allows us to apply Theorem~\ref{thm:readk}.    Using appropriate data structures,  Algorithm~\ref{alg:2cc} runs in $O(n^2(n+m))=O(n^3\log{n}(2c)^{-L})$. One can improve the run time 
in expectation by sampling $O(\log{n})$ neighbors for each node 
in $O(\log n)$ time instead of $O(\log{n}(2c)^{-L})$ time using a standard sublinear sampling technique that generates geometric random variables between successive successes, see \cite{knuth2007seminumerical,tsourakakis2011triangle}. This results in total expected run time $O(n^3\log{n})$.
Since we use a branching factor of $O(\log n)$ for all except the first two levels of $T_u,T_v$,  we work with the $G(n,p)$ model with $p=\frac{40\log{n}}{n}$ to construct the set of almost edge disjoint paths. (Alternatively, one can think that we start with the larger random graph with more edges, and then in the construction of the almost edge disjoint paths we subsample a smaller collection of edges to use in this stage.)  The diameter of this graph \whp grows asymptotically as $L$ \cite{bollobas1998random} for this value of $p$. We use the  $G(n,\frac{40\log{n}(2c)^{-L}}{n})$ model only in Lemma \ref{lem1} to prove that every node has degree at least 
$5\log{n}(2c)^{-L}$.   

 The following result is well known but we present a proof for completeness.  
 
\begin{algorithm}[t]
\caption{\label{alg:2cc} 2-Correlation Clustering with Noise} 
 \begin{algorithmic} 
\STATE $L \leftarrow \diam$
\STATE Perform $ 20n\log{n} (2c)^{-L}$ queries uniformly at random.
\STATE Let $G(V,E,\f)$ be the resulting graph, $\f:E\rightarrow \{+1,-1\}$ 
 \FOR{each item pair $u,v$} 
 \STATE {$\mathcal{P}_{u,v}=\{P_1,\ldots,P_N\} \leftarrow$ Almost-Edge-Disjoint-Paths($u,v$)}  
 \STATE  $Y_i \leftarrow \prod_{e \in P_i} \f(e)$ for $i=1,\ldots,N$  
 \STATE $Y_{uv} \leftarrow \sum_{P \in \mathcal{P}_{u,v}} Y_P $
  \IF{$Y_{uv} \geq 0$} 
  \STATE {$\f(u) = \f(v)$} 
  \ELSE 
  \STATE{$\f(u)\neq \f(v)$} 
  \ENDIF
 \ENDFOR
\end{algorithmic}
\end{algorithm}

\begin{algorithm}[t]
\caption{\label{alg:edgeDisjointPaths}  Almost-Edge-Disjoint-Paths($u,v$)} 
 \begin{algorithmic}  
\REQUIRE $G(V,E,\f)$, $u,v \in V(G)$ 
\STATE $\epsilon \leftarrow \frac{1}{\sqrt{\log\log{n}}}$
\STATE Using Breadth First Search (BFS) grow a tree $T_u$ starting from $u$ as follows.  
\STATE For the first level of the tree, we choose $4\log{n}(2c)^{-\diam}$ neighbors of $u$. 
\STATE For the rest of the tree we use a branching factor equal to $4\log{n}$ until it reaches depth equal to $\epsilon L$. Similarly, grow a tree $T_v$ rooted at $v$, node disjoint from $T_u$ of equal depth.
\STATE From each leaf $u_i$ ($v_i$) of  $T_u$ ($T_v$)  for $i=1,\ldots, N$ 	grow  node disjoint trees until they reach depth $(\frac{1}{2}+\epsilon)L$ with branching factor $4\log{n}$.  Finally, find an edge between $T_{u_i},T_{v_i}$ 
\end{algorithmic}
\end{algorithm}

\begin{claim} 
\label{claim1} 
Consider a  path $P_{uv}$ between nodes $u,v$ of length $L$. Let $R_{uv}=\prod_{e \in P_{uv}} \f(e)$. Then, 
%\begin{equation}
\begin{align}
\label{eq1} 
\Prob{R_{uv}=1|f(u)=f(v)} = \Prob{R_{uv}=-1|f(u)\neq f(v)}     =\frac{1+(1-2q)^L}{2}  &\nonumber
\end{align}
%\end{equation}
\end{claim} 
 
\begin{proof} 
Here $R_{uv} = 1$ iff $\f$ agrees with the unknown clustering function $f$ on $u,v$. This happens when the number of corrupted edges along that path $P_{uv}$ is even. Let $Z_{uv} \sim Bin(L,q)$
be the number of corrupted edges/sign flips along the $P_{uv}$ path.
Clearly, $\Prob{Z_{uv} \text{~is even}}+ \Prob{Z_{uv} \text{~is odd}}=1$. Also,

\begin{align*}
\hspace{-3mm} (1-2q)^L =  \sum_{k=0}^L {L \choose k} (-q)^k (1-q)^{L-k}  &= 
 \sum_{k=0}^{\lfloor \frac{L}{2} \rfloor } {L \choose 2k} q^{2k} (1-q)^{L-2k} - 
   \sum_{k=0}^{\lfloor \frac{L}{2} \rfloor }  {L \choose 2k+1} q^{2k+1} (1-q)^{L-(2k+1)} =&\\
& \Prob{Z_{uv} \text{~is even}} - \Prob{Z_{uv} \text{~is odd}}. 
\end{align*} 
\medskip

\noindent Therefore $\Prob{Z_{uv} \text{~is even}} = \frac{1+(1-2q)^L}{2}$, and the result follows. 
\end{proof}

\noindent  The next lemma is a direct corollary of the lower tail multiplicative Chernoff bound. 

\begin{lemma} 
\label{lem1} 
Let $G \sim G(n,\frac{40\log{n}}{(2c)^Ln})$ be a random binomial graph. Then \whp all vertices have degree greater than $5\log{n}(2c)^{-L}$. 
\end{lemma} 

\begin{proof} 
The degree $deg(u)$ of a node $u \in V(G)$ follows the binomial distribution $Bin(n-1,\frac{40\log{n}}{(2c)^Ln})$. Set $\delta = \frac{3}{4}$. Then 
\begin{align*}
\Prob{  deg(u) <5 \log{n} (2c)^{-L}  } &< e^{- \frac{\delta^2}{2} 40\log{n} (2c)^{-L}}  \ll n^{-1}.
\end{align*}
Taking a union bound over $n$ vertices gives the result.
\end{proof}

\noindent Now we proceed to our construction of sufficiently enough almost edge-disjoint paths. Our construction is based on  standard techniques in random graph theory \cite{broder1998optimal,dudek2015rainbow,frieze2012rainbow,tsourakakis2013mathematical}, we include the full proofs for completeness.

\begin{lemma}
\label{lem2} 
Let $G\sim G(n,p)$ where $p=\frac{40\log n}{n}$. 
Fix $t\in \field{Z}^+$ and $0<\alpha<1$. Then, \whp there does not exist a subset $S \subseteq [n]$, 
such that $|S| \leq \alpha t L$ and $e[S] \geq |S|+t$. 
\end{lemma}

\begin{proof}
Set $s=|S|$.Then,
\begin{align*}
\Prob{ \exists S: s \leq \alpha t L \text{~~and~~} e[S] \geq s+t } 
&  
\leq \sum_{s \leq \alpha t L} \binom{n}{s} \binom{\binom{s}{2}}{s+t} p^{s+t} \leq  \sum_{s \leq \alpha t L } \bfrac{ne}{s}^s \bfrac{es^2p}{2(s+t)}^{s+t}  \\  
\leq \sum_{s \leq \alpha t L } (e^{2+o(1)}\log{n})^s \bfrac{20es\log{n}}{n}^t  
&\leq  \alpha tL \brac{ (e^{2+o(1)}\log{n})^{ \alpha L} \bfrac{20e \alpha t\log^2{n} }{n\log{\log{n}}}}^t <   \frac{1}{n^{(1-\a-o(1))t}}. 
\end{align*}
\end{proof}

\begin{lemma}
\label{lem3}
Let $T$ be a rooted tree of depth at most $\frac{4L}{7}$ and let $v$ be a vertex not in $T$.  Then with probability $1-o(n^{-3})$, $v$ has at most $10$ neighbors in $T$, i.e., $|N(v) \cap T| \leq 10$.
\end{lemma} 

\begin{proof}
Let $T$ be a rooted tree of depth at most $\frac{4L}{7}$ and let $S$ consist of $v$, the neighbors of $v$ in $T$ plus the ancestors of these neighbors. Set $b =|N(v) \cap T|$. Then 
$|S|\leq 4bL/7+1\le 3bL/5$ and $e(S)=|S|+b-2$. It follows from Lemma \ref{lem2} with $\a=3/5$ and $t=8$, that 
we must have $b\leq 10$ with probability $1-o(n^{-3})$.
\end{proof} 

\noindent We show that by growing trees iteratively we can construct 
sufficiently many edge-disjoint paths for $n$ sufficiently large.

\begin{lemma}
\label{lem4} 
Let $\epsilon = \frac{1}{\sqrt{ \log \log n}}$, and $k=\epsilon L$.
For all pairs of  vertices $x,y \in [n]$ there exists a subgraph $G_{x,y}(V_{x,y},E_{x,y})$ of $G$ as shown in figure~\ref{fig:fig1}, {\it whp}.
The subgraph consists of two isomorphic vertex disjoint trees $T_x,T_y$ rooted at $x,y$ each of depth $k$.
$T_x$ and $T_y$ both have a branching factor of  $4\log n(2c)^{-L}$ for the first level, and  $4\log n$ for the remaining levels.
If the leaves of $T_x$ are $x_1,x_2,\ldots,x_\tau,\tau\geq (2c)^{-L}n^{4\epsilon/5}$ then $y_i=f(x_i)$ where $f$ is a natural 
isomorphism.
Between each pair of leaves $(x_i,y_i),i=1,2,\ldots,m$ 
there is a path $P_i$ of length  $(1+2\epsilon) L$. The paths $P_i,i=1,2,\ldots,\tau,\ldots$ are edge disjoint.
\end{lemma}

\begin{proof}
Because we have to do this for all pairs $x,y$, we note without further comment that likely (resp. unlikely) 
events will
be shown to occur with probability $1-o(n^{-2})$ (resp. $o(n^{-2}$)).

To find the subgraph shown in Figure~\ref{fig:fig1}(b) we grow tree structures as shown 
in Figure~\ref{fig:fig1}(a). 
Specifically, we first grow a tree from $x$ using BFS until it reaches  depth $k$. 
Then, we grow a tree starting from $y$ again using BFS until it reaches depth $k$. For the first level of both trees, we choose $\frac{5\log n}{(2c)^L}$ neighbors of $x,y$ respectively. For all other levels we use a branching factor equal to $4\log n$. Before we show how to continue our construction,  we need to prune down the degree of $G([n]\backslash\{x,y\})$ so that the remaining subgraph behave as $G(n,p)$ with $p=\Theta(\frac{\log n}{n})$. This can be achieved  for example  either by considering a random subgraph of $G$ according to $G([n]\backslash\{x,y\}, \frac{40\log n}{n})$, applying Chernoff bounds as in Lemma~\ref{lem1} to show that each node has degree at least $5\log n$, or by letting each node choose $5 \log{n}$ neighbors uniformly at random.  

Finally, once trees $T_x,T_y$ have been constructed, we grow trees from the leaves of $T_x$ and $T_y$ using BFS for depth $\gamma=(\frac{1}{2}+\epsilon)L$. 
Now we analyze these processes. Since the argument is the same we explain
it in detail for $T_x$ and we outline the differences for the other trees. 
We use the notation $D_i^{(\r)}$ for the number of vertices at depth $i$
of the BFS tree rooted at $\r$.

First we grow $T_x$. As we grow the tree via BFS from a vertex $v$ at depth $i$ to vertices 
at depth $i+1$ certain {\em bad} edges from $v$ may point 
to vertices already in $T_x$. Lemma \ref{lem3} shows with probability $1-o(n^{-3})$ there can be at most
10 bad edges emanating from $v$.

Hence, we obtain the recursion  
\beq{rec1}
D_{i+1}^{(x)} \geq \brac{5\log{n}-10} (D_i^{(x)}-1) \geq 4 \log{n}D_i^{(x)}.
\eeq

\noindent Therefore the number of leaves satisfies 

\beq{rec2}
D_{k}^{(x)} \geq  \frac{1}{(2c)^L} \Big(4\log n\Big)^{\e L}\geq \frac{1}{(2c)^L}n^{4\e/5}.
\eeq
We can make the branching factors exactly $4\log n(2c)^{-L}$ for the first level and $4\log n$ for all remaining levels  by pruning.
We do this so that the trees $T_x$ are isomorphic to each other. With a similar argument 
\beq{rec3}
D_{k}^{(y)} \geq \frac{1}{(2c)^L} n^{\frac{4}{5}\epsilon}.
\eeq
The only difference is that now we also say an edge is bad if the other endpoint is in $T_x$.
This immediately gives
$$D_{i+1}^{(y)} \geq \brac{5\log{n}-20} (D_i^{(y)}-1) \geq 4\log{n}D_i^{(y)}$$
and the required conclusion \eqref{rec3}.

Similarly, from each leaf $x_i \in T_x$ and $y_i \in T_y$ we grow trees $\hT_{x_i},\hT_{y_i}$ of depth 
$\gamma = \big(\frac{1}{2}+\epsilon\big) L$ using the same procedure and arguments 
as above. Lemma~\ref{lem3} implies that there are at most 20 edges from the vertex $v$ being explored to 
vertices in any of the trees already constructed (at most 10 to $T_x$ plus any trees rooted at an $x_i$ 
and another 10 for $y$).
The number of leaves of each $\hT_{x_i}$ now satisfies
$$\hD_\g^{(x_i)}\geq (4\log{n})^{\g+1}\geq n^{\frac{1}{2}+\frac{4}{5}\epsilon}.$$
The result is similar for $\hD_\g^{(y_i)}$.

\noindent Observe next that BFS does not condition on the edges between the leaves $X_i,Y_i$ of the trees $\hT_{x_i}$ 
and $\hT_{y_i}$.
That is, we do not need to look at these edges in order to carry out our construction. 
On the other hand we have conditioned on the occurrence of certain events to imply a certain growth rate.
We handle this technicality as follows. We go through the above construction and halt if ever we find that we cannot
expand by the required amount. Let ${\bf A}$ be the event that we do not halt the construction
i.e. we fail the conditions of Lemmas \ref{lem2} or \ref{lem3}. We have
$\Prob{{\bf A}}=1-o(1)$ and so,

\begin{align*}
\Prob{\exists i:e(X_i,Y_i)=0\mid {\bf A}} &\leq \frac{\Prob{\exists i:e(X_i,Y_i)=0}}{\Pr({\bf A})}  
\leq 2n^{\frac{4\e}{5}}(1-p)^{n^{1+\frac{8\e}{5}}} \leq n^{-n^\e}.
\end{align*}

We conclude that {\em whp} there is always an edge between each $X_i,Y_i$ and thus a path of length at most
$(1+2\e)L$ between each $x_i,y_i$.
\end{proof}
 
\begin{figure}[t]
\centering
\begin{tabular}{@{}c@{}@{\ }c@{}} \includegraphics[width=0.45\textwidth]{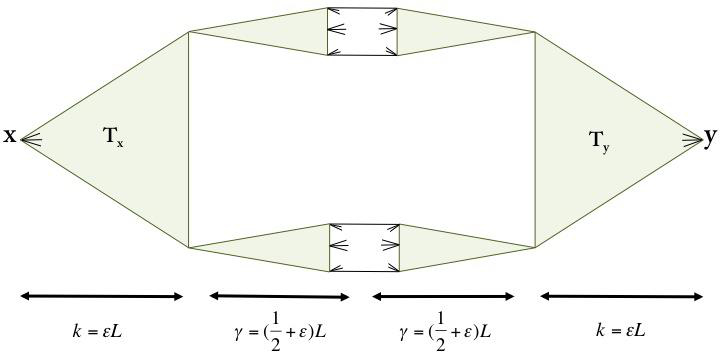}  \hspace{15mm} & \includegraphics[width=0.45\textwidth]{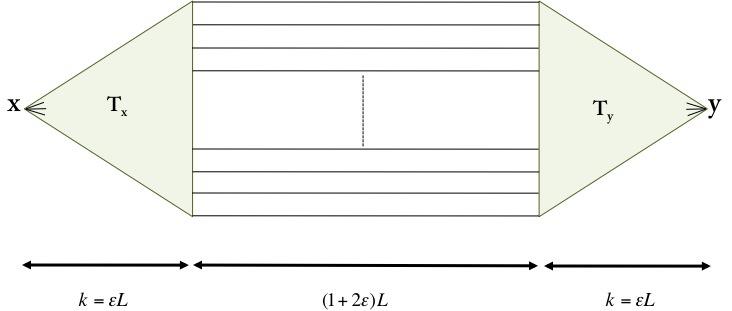} \\
(a)  \hspace{10mm}  & (b)  
\end{tabular}
\caption{\label{fig:fig1} Illustration of construction in Lemma~\ref{lem4}. (a) By repeatedly growing trees  appropriately, (b) we create for each pair of nodes $x,y$ two node disjoint trees $T_x,T_y$ of depth $k= \epsilon L$ whose leaves can be matched via a natural isomorphism and linked with edge disjoint paths of length $(1+o(1))L$}
\vspace{-4mm}
\end{figure}

\noindent The proof of Theorem~\ref{thm:thrm1} follows. Set  $q=\frac{1}{2}-c$.

\begin{proof}[Proof of Theorem~\ref{thm:thrm1}]
Fix a pair of nodes $x,y \in V(G)$.  Let $Y_1,\ldots,Y_N$  be the signs of the $N$ edge  disjoint paths connecting them, i.e., $Y_i \in \{-1,+1\}$ for all $i$. Also let $Y= \sum_{i=1}^N Y_i$. Notice that $\{Y_1,\ldots,Y_N\}$ is a read-$k$ family where $k=\frac{N}{4\log{n}(2c)^{-L}}$.

% Suppose $f(x)=f(y)$. The case $f(x)\neq f(y)$ is worked out in an identical way and details are omitted.   Set  $r = \frac{1+(1-2q)^L}{2}$.
   By the linearity of expectation 
  
  $$\Mean{Y} = N(2c)^L\geq n^{\tfrac{4}{5}\epsilon} (2c)^{L}.$$

\noindent By applying Theorem~\ref{thm:readk} we obtain  

\begin{align*}
\Prob{Y<0} &= \Prob{Y-\Mean{Y} < -\Mean{Y} } \leq \exp \left( -\frac{n^{4/5\epsilon} (2c)^L}{\frac{2n^{4/5\epsilon}  }{4(2c)^{-L}\log{n}}} \right)  =o(n^{-2}).
\end{align*}
\end{proof}

\spara{Planted bisection model.}  Before we prove Theorem~\ref{thm:thrm2} we discuss the connection between our formulation and the well studied planted bisection model. Suppose that $n$ is even, and the graph has two clusters of equal size. 
The probabilities of connecting are $p$ within each cluster, and $q<p$ across the clusters. Now recall  our setting as described in Model \rom{1}, and consider just the edges that correspond to queries that return $+1$. These form a graph drawn from the planted bisection model where $p=\frac{1+c}{2}\times \frac{40 \log n}{n}, q=\frac{1-c}{2}\times \frac{40\log n}{n}$. Therefore, one can apply existing methods for exact recovery, e.g., 
\cite{abbe2016exact,mcsherry2001spectral} instead of our method when the sizes of the two clusters are (roughly) equal. It is worth noting that despite the wide variety of techniques that appear in the context of the planted partition model, including the EM algorithm \cite{snijders1997estimation}, spectral methods \cite{mcsherry2001spectral,vu2014simple}, semidefinite programming \cite{abbe2016exact,hajek2016achieving,montanari2015semidefinite}, hill-climbing \cite{carson2001hill}, Metropolis algorithm \cite{jerrum1998metropolis}, modularity based methods \cite{bickel2009nonparametric}, our  edge-disjoint path technique is novel in this context. 

Hajek, Wu, and Xu proved that when each cluster has $\Theta(n)$ nodes, the average degree has to scale as $\frac{\log n}{ (\sqrt{1-q} - \sqrt{q})^2 }$ for exact recovery \cite{hajek2016achieving}. Also, they showed that using semidefinite programming (SDP) exact recovery is achievable at this threshold \cite{hajek2016achieving}. 
Note that as $q \rightarrow \frac{1}{2}$,  this lower bound scales as $\Theta(\frac{\log n}{ (1-2q)^2 }) = \Theta( \frac{\log n}{(2c)^2})$.  It is an interesting theoretical problem to explore if the techniques we develop in this work, or similar techniques can get closer to this lower bound.

%\spara{Proof of Theorem~\ref{thm:thrm2}.} 

\begin{proof}[Proof of Theorem~\ref{thm:thrm2}]
Since the proof of Theorem~\ref{thm:thrm2} overlaps with the proof of Theorem~\ref{thm:thrm1}, we outline the main differences. Let us return to the basic version of Model \rom{2}, and let $X(e) \in \{-1,0,1\}$
for $e=(x,y)$ be 
$$\tilde{f}(e) - (g(x) - g(y)) \bmod k.$$
\noindent Then given a path between two vertices $u$ and $v$, 
$$g(v) = g(u) + \sum_{e \in P_{uv}} \tilde{f}(e) - \sum_{e \in P_{uv}} X(e) \bmod k.$$
Our question is now what is $Z_{uv} = \sum_{e \in P_{uv}} X(e) \bmod k$.  
We would like that
$Z_{uv}$ be (even slightly) more highly concentrated on 0 than on other
values, so that when $g(u) = g(v)$, we find that the sum of the return
values from our algorithm, 
$\sum_{e \in P_{uv}} \tilde{f}(e) \bmod k$, is most likely to be 0.
We could then conclude by looking over many almost edge-disjoint paths that if this sum is 0
over a plurality of the paths, then $u$ and $v$ are in the same group \whp. 

For our simple error model, the sum $\sum_{e \in P_{uv}} X(e) \bmod k$
behaves like a simple lazy random walk on the cycle of values modulo
$k$, where the probability of remaining in the same state at each step
is $q$.  Let us consider this Markov chain on the values modulo $k$;
we refer to the values as states.  
Let $p_{ij}^{t}$ be the probability of going from state $i$
to state $j$ after $t$ steps in such a walk.  It is well known than
one can derive explicit formulae for $p_{ij}^t$;  see e.g.  \cite[Chapter XVI.2]{feller1968introduction}.  It also follows by simply finding the eigenvalues and
eigenvectors of the matrix corresponding to the Markov chain and using 
that representation.  One can check the resulting forms to determine that 
$p_{0j}^{t}$ is maximized when $j=0$, and to determine the corresponding gap
$\max_{j\in[1,k-1]} |p_{00}^{t}-p_{0j}|^{t}$.  Based on this gap, we can apply
Chernoff-type bounds as in Theorem~\ref{thm:readk} to show that the plurality of almost edge-disjoint
paths will have error 0, allowing us to determine whether the endpoints of the path
$x$ and $y$ are in the same group with high probability.  

The simplest example is with $k=3$ groups, where we find
$$p_{00}^{t} = \frac{1}{3} + \frac{2}{3}\left (1-3q/2 \right )^t,$$
and 
$$p_{01}^{t} = p_{02}^{t} = \frac{1}{3} - \frac{1}{3}\left (1-3q/2 \right )^t.$$
In our case $t = L$, and we see that for any $q < 2/3$, $p_{00}^{t}$ is large enough
that we can detect paths using the same argument as in Model \rom{1}.

For general $k$, we use that the eigenvalues of the matrix
\[
\begin{bmatrix}
    1-q & q & 0 & \dots & q    \\
    q   & 1-q & q & \dots & 0    \\
    \vdots & \vdots & \vdots & \vdots & \ddots \\
    q & 0 & 0 & \dots  & 1-q  
\end{bmatrix}
\]
are 
$1-q+q \cos(2\pi j/k), j = 0, \ldots, k-1$, with the $j$-th corresponding 
eigenvector being $[1, \omega^j, \omega^{2j}, \ldots, \omega^{j(k-1)}]$ where
$\omega = e^{2\pi i/k}$ is a primitive $k$th root of unity.  Here, $i$ is not an index but the square root of -1, i.e., $i=\sqrt{-1}$.
In this case we have
$$p_{00}^{t} = \frac{1}{k} + \frac{1}{k}\sum_{j=1}^{k-1} \big(1-q +q \cos(2\pi j/k)\big)^t.$$
Note that $p_{00}^{t} > 1/k$.  
Some algebra reveals that the next largest value of $p_{0j}^{t}$ belongs to $p_{01}^{t}$, 
and equals
$$p_{01}^{t} = \frac{1}{k} + \frac{1}{k}\sum_{j=1}^{k-1} \omega^{-j} \big(1-q +q \cos(2\pi j/k)\big)^t.$$
We therefore see that the error between ends of a path again have the plurality value 0, with a gap of at least 
$$p_{00}^{t} - p_{01}^{t} \geq 2(1-\cos(2\pi/k))(1-q +q \cos(2\pi/k))^t.$$
This gap is constant for any constant $k \geq 3$ and $q \leq 1/2$.
\end{proof}

\noindent As we have already mentioned, the same approach could be used for the more general setting where
$$\tilde{f}(e) = g(x)-g(y)+j \text{~~with probability~~} q_j, 0 \leq j< k,$$
but now one works with the Markov chain matrix 
\[
\begin{bmatrix}
    q_0 & q_1 & q_2 & \dots & q_{k-1} \\
    q_{k-1}   & q_0 & q_1 & \dots & q_{k-2} \\
    \vdots & \vdots & \vdots &\ddots & \vdots \\
    q_1 & q_2 & q_3 & \dots & q_0
\end{bmatrix}.
\]

 \section{Related Work}
\label{sec:related}
Fritz Heider introduced the notion of a signed graph, with $+1$ or $-1$ labels on the edges, 
 in the context of balance theory \cite{heider1946attitudes}.  The key subgraph in balance theory is the {\em triangle}: any set of three fully interconnected nodes whose product of edge signs is negative is not  balanced. The  complete graph is balanced if every one of its triangles
is balanced.  Early work on signed graphs focused on graph theoretic properties of balanced graphs  \cite{cartwright1956structural}.  Harary proved the famous balance theorem which characterizes balanced graphs \cite{harary1953notion}. 

Bansal et al. \cite{bansal2004correlation} studied  \cc: given an undirected signed graph  partition the nodes into clusters so that the total number of disagreements is minimized. This problem is NP-hard \cite{bansal2004correlation,shamir2004cluster}. Here, a disagreement can be either a positive edge between vertices in two clusters or a negative edge between two vertices in the same cluster. Note that in \cc the number of clusters is not specified as part of the input. The case when the number of clusters is constrained to be at most two is known as \cccc.  

Minimizing disagreements is equivalent to maximizing the number of agreements. However, from an approximation perspective these two versions are different: minimizing is harder.  For minimizing disagreements, Bansal et al. \cite{bansal2004correlation} provide a 3-approximation algorithm for \cccc, and Giotis and Guruswami provide a polynomial time approximation scheme (PTAS) \cite{giotis2006correlation}. Ailon et al. designed a 2.5-approximation algorithm \cite{ailon2008aggregating} that was further improved by  Coleman et al. to a 2-approximation \cite{coleman2008local}. 
We remark that the notion of {\em imbalance} studied by Harary is  the \cccc cost of the signed graph.  Mathieu and Schudy initiated the study of noisy correlation clustering \cite{mathieu2010correlation}. They develop various algorithms when the graph is complete,  both for the cases of a random and a semi-random model. Later, 
Makarychev,   Makarychev,  and Vijayaraghavan proposed an algorithm for graps with $O(n\text{poly}\log n)$ edges under a semi-random model \cite{makarychev2015correlation}. 

When the graph is not complete \cc and {\sc Minimum Multicut} reduce to one another leading to a $O(\log{n})$ approximations \cite{charikar2003clustering,demaine2003correlation}. The case of constrained size clusters has recently been studied by Puleo and Mileknovic \cite{puleo2015correlation}.   Finally, by using the Goemans-Williamson SDP relaxation for {\sc Max Cut} \cite{goemans1995improved}, one can obtain a 0.878 approximation guarantee for \cccc problem as noted by \cite{coleman2008local}. 
 
Chen et al. \cite{chen2014clustering,chen2012clustering}  consider also model \rom{1} as described in Section~\ref{sec:introduction} and provide a method that can reconstruct the clustering for random binomial graphs with $O(n \text{poly} \log n$ edges. Their method exploits low rank properties of the cluster matrix, and requires certain  conditions, including conditions on the imbalance between clusters, see \cite[Theorem 1, Table 1]{chen2012clustering} to be true.  Their methods is based on a convex relaxation of a low rank problem. 
Also, closely related to our work lies the work of  Cesa-Bianchi et al. \cite{cesa2012correlation} that take a learning-theoretic perspective on the problem of predicting signs. They consider three types of models: batch, online, and active learning, and provide theoretical bounds for prediction mistakes for each setting. They use the correlation clustering objective as their learning bias, and they show that the risk of the empirical risk minimizer is controlled by the correlation clustering objective.   Chian et al. point out that the work of Cand{\`e}s and Tao \cite{candes2006robust} can be used to predict signs of edges, and also provide various other methods, including singular value decomposition based methods, for the sign prediction problem \cite{chiang2014prediction}. The incoherence is the key parameter that determines the number of queries, and is equal to the group imbalance $\tau = \max\limits_{\text{cluster~} C} \frac{n}{|C|}$. The number of queries needed for exact recovery under Model \rom{1} is $O(\tau^4 n \log^2{n})$.   

%\section{Experimental results}
%\label{sec:exp}
%\input{Experiments.tex}	

\section{Conclusion}
\label{sec:concl}
In this work we have studied the edge sign prediction problem, showing that under our proposed correlation clustering formulation and a fully random noise model querying $O(n^{1+o(1)}\log{n})$ pairs of nodes uniformly at random suffices to recover the clusters efficiently, \whp.  
We also provided a generalization of our model and proof approach to more than two clusters.  
While our work here is theoretical, in future work we plan to apply our method and additional heuristics to real data, and compare against related alternatives.   From a theoretical perspective, it is an interesting  problem to close the gap between our upper bound 
and the known  $\frac{\log n}{(2c)^2}$  lower bound for exact recovery \cite{hajek2016achieving} using techniques based on BFS.

\section*{Acknowledgment}

We would like to thank Bruce Hajek  and  Zeyu Zhou for detecting an error in an earlier version of our work. We also want to thank Yury Makarychev, Konstantin Makarychev, and Aravindan Vijayaraghavan for their feedback.

This work was supported in part by NSF grants CNS-1228598, CCF-1320231, CCF-1563710, and CCF-1535795.

%  \bibliographystyle{abbrv}
%\bibliography{ref}

    \printbibliography

\end{document}